\documentclass[runningheads]{llncs}

\usepackage[utf8]{inputenc}
\usepackage[ruled,linesnumbered]{algorithm2e}
\usepackage{amssymb, amsfonts, amsmath}
\usepackage{stmaryrd}
\usepackage[inline]{enumitem}
\usepackage{tikz}
\usetikzlibrary{automata, positioning, shapes.multipart}
\usepackage{float}
\usetikzlibrary{automata, positioning, arrows.meta}
\newcommand{\semantics}[1]{{\llbracket #1 \rrbracket}}

\newcommand{\R}{\mathbb{R}}
\newcommand{\Stoch}{\mathcal{S}}

\begin{document}

\title{From Transformers to Weighted Automata: Towards the Verification of Large Language Models}
\titlerunning{From Transformers to Weighted Automata}

\author{
    Smayan Agarwal\inst{1} \and
    Aslah Ahmad Faizi\inst{2} \and
    Shobhit Singh\inst{1} \and
    Aalok Thakkar\inst{1}
}

\institute{
    Ashoka University, India\\
    \email{$\{$smayan.agarwal\_ug25, shobhit.singh, aalok.thakkar$\}$@ashoka.edu.in}
    \and
    Independent Researcher\\
    \email{aslahahmadfaizi@gmail.com}
}

\authorrunning{S. Agarwal et al.}

\maketitle

\begin{abstract}
Large language models (LLMs) are increasingly deployed in safety-critical settings, yet their black-box nature makes it difficult to provide formal guaranties about their behavior. Existing verification approaches rely primarily on empirical probing and testing, leaving open the question of how to reason rigorously about general-purpose transformer architectures.

In this work, we establish a principled bridge between transformers and weighted automata, a classical model from formal language theory. This connection enables us to transfer verification tools from automata theory to the analysis of LLMs. Our contributions are twofold: First, we develop a formal correspondence between transformer architectures and weighted automata over reals, showing how distributional properties of LLMs can be captured within this framework. Second, we introduce an identity testing algorithm for weighted automata that provides a statistical method for distinguishing whether two stochastic models define the same distribution up to a tolerance threshold.

This work provides the first formal bridge between modern neural sequence models and classical automata theory, clarifying both the potential and the computational challenges for rigorous LLM verification.

\keywords{Large Language Models, Transformers, Weighted\\ Automata, Verification, Identity Testing, Formal Methods}
\end{abstract}

\section{Introduction}
\label{sec:intro}

The success of transformer-based language models has created an urgent need for principled methods to understand and verify their behaviour. As these models are deployed in safety-critical applications, ranging from medical diagnosis to legal reasoning, the black-box nature of their decision-making processes poses fundamental challenges for certification and reliability guarantees.

Theoretical results from formal language theory offer an alternative approach. By establishing systematic correspondences between transformer variants and well-characterized computational models, we can leverage automata theory, circuit complexity, and descriptive logic to obtain predictive accounts of capabilities and limitations for specific architectural choices.

\subsection{The Gap Between Benchmarks and True Capabilities}

Large Language Models (LLMs) perform well on benchmarks, and yet fail on general tasks. Recent work demonstrates that even state-of-the-art models collapse when faced with simple inductive reasoning problems \cite{Bowen2024,Hua2025}. Furthermore, these models often rely on superficial neighbour-based heuristics rather than generalization \cite{Li2024}.

Current evaluation practices systematically inflate performance estimates while obscuring true model capabilities. A critical issue is \emph{benchmark data contamination}, where test items appear in training corpora. Studies show that GPT-4 achieves over 50\% accuracy on contaminated MMLU items, fundamentally violating the train–test split principle \cite{Deng2023,Xu2024}. This contamination operates alongside preference leakage, where evaluator LLMs exhibit systematic bias toward models sharing training data or architectural relationships, creating circular evaluation systems that measure memorization rather than understanding.

Recent frameworks expose these limitations. Cohen-Inger et al.'s C-BOD analysis demonstrates that high-performing models rely on ``superficial memorization of benchmark cues'' \cite{CohenInger2025}. Similarly, Shojaee et al.'s examination of Large Reasoning Models reveals ``complete accuracy collapse beyond certain complexities'' and ``counterintuitive scaling limits'' where reasoning effort paradoxically decreases on challenging problems despite adequate computational resources \cite{Shojaee2025}. These highlight a need for formal guarantees.

\subsection{The Verification Challenge}

Formal verification of neural models faces a fundamental computational barrier. 
For a model 
$M:\Sigma^\ast \to \mathcal{Y}$, and a property $R$, verification requires establishing the universal claim $\forall w \in D: R(M(w)) = 1$. This is only tractable when the input domain $D$ admits symbolic specification (e.g., as a regular language or logical constraint). When $D$ is defined extensionally by length bounds $D = \Sigma^{\leq n}$, the domain size becomes $\Theta(|\Sigma|^n)$. 
Any sound black-box verifier then faces a worst-case $\Omega(|D|)$ query lower bound, making universal verification computationally infeasible for even modest $n$.

Practical verification therefore requires structured domains $D$ that admit finite abstraction through techniques like product constructions or model checking, rather than exhaustive enumeration of all token sequences up to length $n$.

\subsection{Why Global Constraints Are Insufficient}

While enforcing global constraints, such as $K$-Lipschitz bounds, appears attractive for certification, this approach fundamentally mismatches capability requirements. Lipschitz constraints depend on continuous embedding metrics rather than discrete string functions, compound pessimistically across layers, and force low-contrast attention that prevents the sharp routing needed for algorithmic behaviors. Empirically, global Lipschitz training trades accuracy for minimal certified radii, over-regularizing the selective attention mechanisms that give Transformers their power while providing certificates too weak to justify the capability loss \cite{huster2019limitations,zhang2022rethinking,ducotterd2024improving,verine2022expressivity,delattre2023lvm,khromov2024fundamental,qi2023lipsformer,newhouse2025training}.

\subsection{Automata-Theoretic Foundations}

Recent theoretical advances have begun mapping transformer expressivity to classical computational models. 
Notably, Angluin et al. \cite{angluin2023starfree} proved that rightmost hard-attention encoders with strict causal masking recognize exactly the star-free languages (first-order logic over strings). Adding periodic positional predicates upgrades this class to FO[MOD], while arbitrary monadic predicates yield FO[Mon]. These precise equivalences reveal that transformers align with well-studied fragments of logic and automata theory.
Similarly, architectural variants map naturally to circuit complexity classes: unique and averaging hard-attention encoders correspond to $AC^0$ and $TC^0$ respectively \cite{hao2022formal,strobl2023survey}. Such characterizations suggest that transformer expressivity analysis can be grounded in established frameworks rather than relying solely on empirical exploration.

\subsection{Weighted Finite Automata as an Abstraction Layer}

While the above results apply to restricted transformer variants, modern LLMs use soft attention mechanisms that compute distributions over sequences. This motivates our focus on Weighted Finite Automata (WFA) over reals as a natural model for stochastic sequence generation that generalizes both classical finite automata and hidden Markov models.
WFA provide several advantages for LLM analysis:

\begin{enumerate}
    \item Compositional structure: Unlike neural networks, automata admit modular analysis through product constructions and closure properties.
    \item Algorithmic foundation: Classical algorithms exist for equivalence testing, minimization, and property verification.
    \item Statistical learning theory: Spectral methods enable provable recovery of automata from sample data \cite{balle2012spectral,hsu2012spectral}.
\end{enumerate}

\subsection{Technical Approach and Contributions}

The key two insights of our work are reflected in the two-step reduction: 

\paragraph{Step 1: Transformer to WFA.} We prove that the self-attention mechanism in decoder-only transformers can be exactly computed by a structured RNN that maintains running accumulators for the attention normalization. We then apply spectral learning to the RNN's induced function, using Hankel matrix analysis to extract a finite-state weighted automaton approximation. We discuss this in Section~\ref{sec:approx}.

\paragraph{Step 2: Identity Testing of WFA.} We develop the first identity testing procedure for rational stochastic languages on infinite domains. This technique can be applied to the WFA obtained in Step 1 to provide formal guarantees about a restricted transformer architecture. We discuss this in Section~\ref{sec:testing}.

Through these, our work makes two primary contributions toward verification and explainability of LLMs:
\begin{enumerate}
    \item \textbf{Formal reduction theorem:} We prove that decoder-only transformers are equivalent to structured RNNs, which can be approximated by WFA with controlled error bounds.
    \item \textbf{A black-box algorithm for identity testing.} We design an algorithm that, given only sample access to an unknown distribution over strings (e.g., LLM outputs), tests whether this distribution is $\varepsilon$-close to some WA with confidence $1-\delta$. This test serves as a \emph{certificate of abstraction}: whenever it accepts, one can meaningfully reason about the black-box model via the corresponding WA. This enables (a) verification of safety and correctness properties, (b) improved interpretability through automaton-based explanations of sequential behavior, and (c) applicability to proprietary black-box models where internal parameters are inaccessible.
\end{enumerate}

\section{Notation and Preliminaries}
\label{sec:notation}

This section establishes the notation and foundational concepts used throughout the paper. We assume familiarity with basic formal language theory, linear algebra, and probability theory.

\subsection{Sets, Alphabets, and Strings}

Let $\Sigma$ denote a finite alphabet. In the context of LLMs, the alphabet $\Sigma$ typically represents the model's \textit{vocabulary} or \textit{tokenizer}, where each symbol $\sigma \in \Sigma$ corresponds to a unique token. The set of all finite strings over $\Sigma$ is denoted $\Sigma^*$, which includes the empty string $\epsilon$. The length of a string $x$ is denoted $|x|$. The set of all strings of length at most $\theta$ is denoted $\Sigma^{\leq \theta}$.

\subsection{Distributions and Distances}

A distribution $P$ over a countable set $\mathcal{X}$ (in our case, $\mathcal{X} = \Sigma^*$) is a function $P: \mathcal{X} \to [0, 1]$ such that $\sum_{x \in \mathcal{X}} P(x) = 1$. We denote the set of all distributions over $\Sigma^*$ by $\Stoch(\Sigma^*)$, and call them stochastic languages.

The primary distance measure between two distributions $P$ and $Q$ is the $\ell_1$ distance (which is twice the total variation distance):
\[
\|P - Q\|_1 = \sum_{x \in \Sigma^*} |P(x) - Q(x)|.
\]
For a distribution $P$ and a set $S \subseteq \Sigma^*$, we write $P(S) = \sum_{x \in S} P(x)$.

\subsection{Weighted Finite Automata and Rational Stochastic Languages}
\begin{definition}[Weighted Finite Automaton (WFA)]\label{def:weighted}
Let $\Sigma$ be a finite alphabet. For a given semiring $(K,+,\cdot,0,1)$, a weighted finite automaton is a tuple $\mathcal{A} = (\Sigma, Q, \lambda, \mu, \rho)$,
where $Q$ is a finite set of states, $\lambda,\mu \in K^Q$ are the initial and final weight functions, and $\rho \colon Q \times \Sigma \times Q \to K$ is the transition weight function.  The semantics of $\mathcal{A}$ is the function $f_\mathcal{A} \colon \Sigma^* \to K$ given by
\[
f_\mathcal{A}(w) \;=\; \sum_{q_0,\dots,q_n \in Q}
  \lambda(q_0) \cdot \rho(q_0,w_1,q_1) \cdots \rho(q_{n-1},w_n,q_n) \cdot \mu(q_n),
\]
for $w = w_1 \cdots w_n \in \Sigma^*$. 
\end{definition}

\begin{figure}[t]
    \centering
    \begin{tikzpicture}[->,auto,node distance=3cm,accepting/.style=accepting by arrow]
        % States
        \node[state,initial, initial where=above,accepting,accepting text={$\tfrac{1}{4}$},accepting where=below] (q1) {$q_1$};
        \node[state] (q2) [right=of q1] {$q_2$};

        % Transitions
        \path (q1) edge [loop left] node {$a : \tfrac{1}{4}$} ()
        (q1) edge [bend left] node {$a : \tfrac{1}{4}; b : \tfrac{1}{4}$} (q2)
              (q2) edge [loop right] node {$b : \tfrac{1}{3}$} ()
              (q2) edge [bend left] node {$a : \tfrac{2}{3}$} (q1);
    \end{tikzpicture}
    \caption{A weighted finite automaton over the alphabet $\{a, b\}$. 
    State $q_1$ is both initial and accepting with weight $1$. 
    The transition labels indicate the symbol and the weight (probability) associated with each move.}
    \label{fig:notation:weighted-automata}
\end{figure}
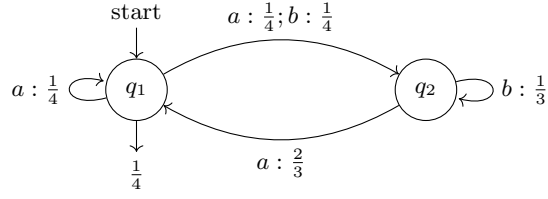

\begin{definition}[Rational Stochastic Language]\label{def:stochastic-rational}
The class of rational stochastic languages over alphabet $\Sigma$ and field $\R$, denoted $\Stoch_{\R}^{rat}(\Sigma)$, consists of all functions $f \colon \Sigma^+ \to \R$ such that:

\begin{enumerate}
    \item $f = \semantics{\mathcal{A}}$ for some WFA $\mathcal{A}$ over the semiring $(\R, +, \cdot, 0, 1)$
    \item $f$ is a stochastic language 
\end{enumerate}

The restricted class $\Stoch_{\R^+}^{rat}(\Sigma)$ considers only automata with non-negative weights. When the semiring is clear from context, we write simply $\Stoch^{rat}(\Sigma)$.
\end{definition}

Figure~\ref{fig:notation:weighted-automata} shows a WFA over the alphabet $\{a, b\}$. It induces a rational stochastic language over $\{a, b\}^*$. For a detailed discussion on WFA and stochastic languages, see \cite{bollig2015weighted}. 

\subsection{Recurrent Neural Networks}

\begin{definition}[Recurrent Neural Network (RNN)]
A \emph{recurrent neural network (RNN)} is a triple \(R=(\alpha, g, \beta)\) where:
\begin{enumerate*}[label=(\roman*)]
    \item \(\alpha \in \mathbb{R}^d\) is the initial state vector,
    \item \(g: \mathbb{R}^d \times \Sigma \to \mathbb{R}^d\) is the recurrent transition function,
    \item \(\beta: \mathbb{R}^d \to \mathbb{R}\) is the output function.
\end{enumerate*}
The RNN processes a sequence \(w = w_1 \dots w_n \in \Sigma^*\) by producing a sequence of hidden states:
\[
\begin{aligned}
h_0 &= \alpha, \\
h_i &= g(h_{i-1}, w_i) \quad \text{for } i = 1, \dots, n.
\end{aligned}
\]
The output of the RNN for the sequence is \(f_R(w) = \beta(h_n)\). \cite{merrill2020formal}
\end{definition}

Intuitively, an RNN maintains a hidden state \(h_i\) that serves as a compressed memory of all past inputs up to position \(i\). 
At each step, the transition function \(g\) updates this state based on the previous hidden state and the current input symbol.
The initial vector \(\alpha\) represents the network’s starting memory, and the output function \(\beta\) maps the final hidden state to a scalar value or prediction. 
In standard neural implementations, \(g\) and \(\beta\) are parameterized by weight matrices and nonlinearities such as \(\tanh\) or \(\mathrm{ReLU}\), 
allowing the RNN to approximate complex sequence-to-value mappings. 

\subsection{The Spectral Learning Framework}

Spectral learning offers a method to learn a finite state representation from i.i.d. samples \cite{hsu2012spectral,balle2012spectral}.
The \textit{Hankel matrix} $H_f \in \mathbb{R}^{\Sigma^* \times \Sigma^*}$ of a function $f: \Sigma^* \to \mathbb{R}$ is a bi-infinite matrix whose entries are defined for any $(p, s) \in \Sigma^* \times \Sigma^*$ by $H_f[p, s] = f(ps)$, where $ps$ denotes the concatenation of $p$ and $s$. 

A fundamental result \cite{schutzenberger1961definition} states that $f$ is rational (i.e., computable by a WFA of $d$ states) if and only if its Hankel matrix $H_f$ has finite rank $d$. We discuss this in further detail in Section~\ref{sec:approx:spectral}. 

\subsection{Transformer Attention}

Transformers \cite{vaswani2017} are the popular LLM design architecture. At a high level, they operate over sequences of discrete symbols from
a finite alphabet~$\Sigma$, which are mapped to sequences of vectors via an embedding layer. Each transformer layer is a length-preserving function on such vector sequences, combining \emph{self-attention}---which retrieves and mixes information from a variable-length history---and \emph{feed-forward networks} applied position-wise. Stacks of these layers, together with input and output projections, yield a powerful sequence model. 

For theoretical analysis, it is crucial to treat transformers as operating on inputs of \emph{unbounded length}, rather than fixing a maximum context window. This allows us to compare them with formal language models such as automata, circuits, and logics, and to ask: \emph{what languages can transformers recognize or generate?} Variants of the architecture (encoder, decoder, encoder--decoder), choices of attention mechanism (softmax, hard), and restrictions on precision or positional encodings all impact expressivity~\cite{strobl2023survey}.

\begin{definition}[Decoder-Only Transformer]
\label{def:transformer}
Let \(\Sigma\) be a finite alphabet and \(d \in \mathbb{N}\).  
A one-layer decoder-only Transformer is a tuple
\[
\mathcal{T} = (\alpha, W_Q, W_K, W_V, \beta),
\]
where
\begin{itemize}
    \item \(\alpha : \Sigma \to \mathbb{R}^d\) is the embedding map assigning each symbol an embedding vector;
    \item \(W_Q, W_K, W_V \in \mathbb{R}^{d \times d}\) are the learned query, key, and value projection matrices;
    \item \(\beta:\mathbb{R}^d \to \mathbb{R}\) is the output map.
\end{itemize}

Given a sequence \(w = \sigma_1 \cdots \sigma_n \in \Sigma^*\), the model first maps it to embeddings \(e_i = \alpha(\sigma_i)\).  
At each position \(n\), the layer computes a \emph{query} vector \(q_n = W_Q e_n\) for the current symbol, as well as \emph{key} and \emph{value} vectors \(k_i = W_K e_i\), \(v_i = W_V e_i\) for all previous symbols.  
The hidden state at position \(n\) is then obtained by attending to the past:
\[
h_n = \frac{\sum_{i=1}^n v_i \, e^{q_n^\top k_i}}{\sum_{j=1}^n e^{q_n^\top k_j}},
\]
that is, a weighted sum of the values \(v_i\), where the weights depend on how similar the current query \(q_n\) is to each past key \(k_i\).  
Finally, the model output on \(w\) is
\[
f_{\mathcal{T}}(w) := \beta(h_n).
\]
\end{definition}

Intuitively, a decoder-only Transformer processes a sequence one token at a time, much like an RNN, but instead of maintaining a single recurrent hidden state, 
it computes a new hidden representation \(h_n\) for the current token by looking at all previous tokens simultaneously through the \emph{attention mechanism}. 
This mechanism allows Transformers to model long-range dependencies more effectively than RNNs, because each token can directly attend to any previous token without being bottlenecked by a single hidden state.

\section{From Transformers to Weighted Finite Automata}
\label{sec:approx}

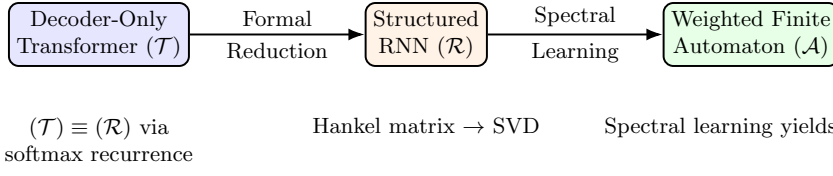
\begin{figure}[t]
\centering
\resizebox{\textwidth}{!}{
\begin{tikzpicture}[node distance=2.5cm,>=Latex, thick]
  \node[draw, rectangle, rounded corners, align=center, fill=blue!10] (T) 
    {Decoder-Only\\ Transformer ($\mathcal{T}$)};
    
  \node[draw, rectangle, rounded corners, align=center, fill=orange!10, right=of T] (RNN) 
    {Structured\\ RNN ($\mathcal{R}$)};
    
  \node[draw, rectangle, rounded corners, align=center, fill=green!10, right=of RNN] (WA) 
    {Weighted Finite\\ Automaton ($\mathcal{A}$)};
    
  \draw[->, thick] 
  (T) -- (RNN) 
  node[above, font=\small, pos=0.5] {Formal} 
  node[below, font=\small, pos=0.5] {Reduction};

  \draw[->, thick] 
  (RNN) -- (WA)
  node[above, font=\small, pos=0.5] {Spectral} 
  node[below, font=\small, pos=0.5] {Learning};
  
  \node[below=0.6cm of T, font=\footnotesize, align=center] (AT1) 
    {$(\mathcal{T}) \equiv (\mathcal{R})$ via \\ softmax recurrence};
    
  \node[below=0.6cm of RNN, font=\footnotesize, align=center] (AT2) 
    {Hankel matrix $\rightarrow$ SVD};
    
  \node[below=0.6cm of WA, font=\footnotesize, align=center] (AT3) 
    {Spectral learning yields WFA};
\end{tikzpicture}
}
\caption{Pipeline for reducing a decoder-only Transformer $\mathcal{T}$ to a weighted automaton $\mathcal{A}$. 
Step~1: establish equivalence with a structured RNN $\mathcal{R}$ (Sec.~\ref{sec:approx:equivalence}). 
Step~2: recover $\mathcal{A}$ from $\mathcal{R}$ via spectral learning on truncated Hankel matrices (Sec.~\ref{sec:approx:spectral}).}
\label{fig:transformer-rnn-wa}
\end{figure}

We reduce a one-layer decoder-only transformer to a WFA via an intermediate RNN.  
Figure~\ref{fig:transformer-rnn-wa} illustrates the pipeline: the transformer $\mathcal{T}$ is first unfolded into an equivalent structured RNN $\mathcal{R}$; then, $\mathcal{R}$ is approximated by a WFA $\mathcal{A}$ through spectral learning.

\subsection{Decoder-Only Transformer as a Recurrent Model}
\label{sec:approx:equivalence}
We now demonstrate an interesting viewpoint: that the self-attention mechanism in transformers can be computed by a carefully constructed RNN \cite{zhang2024}. This equivalence provides a fascinating perspective on the computational relationship between these two fundamental architectures.\\

\paragraph{Construction.}
We construct a specific RNN that computes the softmax self-attention mechanism of a decoder-only transformer at position $n$. Let \(\sigma(x) = (1+e^{-x})^{-1}\) denote the sigmoid function.

For the query vector \(q_n \in \mathbb{R}^d\) and a sequence of key-value pairs \(\{(k_i, v_i)\}_{i=1}^n\), the goal of the RNN is to compute the value accumulator: 

\begin{equation}
\label{eq:approx:accumulator}
f_n = \frac{\sum_{j=1}^n v_j e^{q_n^\top k_j}}{\sum_{j=1}^n e^{q_n^\top k_j}}.
\end{equation}

We define an RNN whose state \(h_i \in \mathbb{R}^{d+1}\) is a composite vector containing two components:
\begin{enumerate}
    \item \(g_i \in \mathbb{R}\): A scalar for the running log-partition function (\(\log \sum_{j=1}^i e^{q_n^\top k_j}\)).
    \item \(f_i \in \mathbb{R}^d\): A vector for the running weighted sum of values (\(\sum_{j=1}^i v_j e^{q_n^\top k_j}\)).
\end{enumerate}
We pack these into a single state vector: $h_i = (g_i, f_i)$.

The initial state, transition function, and output function for this RNN are defined as follows:
\begin{align*}
\alpha &= h_0 = (q_n^\top k_1, v_1), \\
g(h_{i-1}, w_i) &=
\left(
\begin{array}{c}
\log(e^{g_{i-1}} + e^{q_n^\top k_i}) \\
v_i \cdot \sigma(q_n^\top k_i - g_{i-1}) + f_{i-1} \cdot \sigma(g_{i-1} - q_n^\top k_i)
\end{array}
\right), \\
\beta(h_n) &= \frac{f_n}{e^{g_n}}.
\end{align*}
Here, the input \(w_i\) to the RNN at step \(i\) is used to retrieve the corresponding key-value pair \((k_i, v_i)\). The output function \(\beta\) performs the final normalization.

\begin{theorem}[Equivalence of Attention and RNN]
\label{thm:equivalence}
For any input sequence \((e_1,\dots,e_n)\) with associated keys and values, the output of the self-attention mechanism at position \(n\) (Eq.~\ref{eq:approx:accumulator}) is equal to the output of the constructed RNN:
\[
h_n^{\text{attention}} = \beta(g^*(\alpha, (k_1v_1, \dots, k_n v_n))).
\]
\end{theorem}

\begin{proof}[Sketch]
We prove by induction that after processing the \(i\)-th element, the RNN's state contains the correct intermediate values:
\[
e^{g_i} = \sum_{j=1}^i e^{q_n^\top k_j}, \qquad f_i = \sum_{j=1}^i v_j e^{q_n^\top k_j}.
\]
The base case \(i=1\) holds by the definition of $\alpha$. For the inductive step:
\begin{itemize}
    \item The update for \(g_i\) is defined as the log-sum-exp of the previous accumulator and the new term, which correctly computes the new log-partition.
    \item The update for \(f_i\) is a weighted sum. Using the identities for the sigmoid function,
    \[
    \begin{aligned}
    \sigma(q_n^\top k_i - g_{i-1}) &= \frac{e^{q_n^\top k_i}}{e^{g_{i-1}} + e^{q_n^\top k_i}} = \frac{e^{q_n^\top k_i}}{e^{g_i}}, \\
    \sigma(g_{i-1} - q_n^{\top k_i}) &= 
    \frac{e^{g_{i-1}}}{e^{g_{i-1}} + e^{q_n^\top k_i}} = \frac{e^{g_{i-1}}}{e^{g_i}}.
    \end{aligned}
    \]
    Substituting these into the transition function shows that:
    \[
    f_i = \frac{e^{q_n^\top k_i}}{e^{g_i}} v_i + \frac{e^{g_{i-1}}}{e^{g_i}} f_{i-1}.
    \]
    By the inductive hypothesis, $f_{i-1} = \sum_{j=1}^{i-1} v_j e^{q_n^\top k_j}$ and $e^{g_{i-1}} = \sum_{j=1}^{i-1} e^{q_n^\top k_j}$. Therefore, this update rule ensures $f_i = \sum_{j=1}^i v_j e^{q_n^\top k_j}$.
\end{itemize}
The output function \(\beta\) then divides \(f_n\) by \(e^{g_n}\) to produce the final normalized output as in Equation~\ref{eq:approx:accumulator}, which is precisely the output of the attention mechanism, completing the proof.
\end{proof}

\paragraph{Interpretation.} 
This construction shows that the softmax self-attention mechanism in a transformer layer is not inherently non-recurrent. It can be \emph{exactly} computed by an RNN with a carefully designed state and transition function. The RNN's state maintains two accumulators: one for the normalizing denominator and one for the numerator of the attention-weighted sum. This establishes a formal equivalence between the two models for this specific computation.

\subsection{Spectral Learning for RNN Approximation}
\label{sec:approx:spectral}
We now reduce the RNN constructed in the previous section, to a WFA using the spectral learning method \cite{balle2012spectral,balle2020learning,Okudono2019WAE}. The core insight is that the function \(f_R: \Sigma^* \to \mathbb{R}\) computed by the RNN defines a Hankel matrix \(H\) whose rank reveals the intrinsic state complexity of \(f_R\). If this matrix is low-rank, the function can be compactly represented by a WA \cite{schutzenberger1961definition}.

A direct application of the spectral method is infeasible for unbounded sequences. To manage this, we introduce a truncation parameter \(\theta \in \mathbb{N}\), restricting our analysis to the set of sequences of length at most \(\theta\), denoted \(\Sigma^{\leq\theta}\). For any function \(f:\Sigma^*\to\mathbb{R}\), we write \(f_{\leq\theta}\) for its restriction to this finite domain.

\paragraph{Spectral Learning Procedure.}
The procedure for converting the RNN \(R\) to a WA \(A\) is formalized in Algorithm~\ref{alg:rnn-to-wa} and involves the following steps:

\begin{enumerate}
    \item \textbf{Define Prefix and Suffix Sets:} Fix the truncation length \(\theta\). Set the prefix set \(P\) and the suffix set \(S\) to both be \(\Sigma^{\leq \theta}\).
    \item \textbf{Construct the Hankel Matrix:} Form the Hankel matrix \(H \in \mathbb{R}^{P \times S}\), where each entry is defined by the output of the RNN:
    \begin{equation}
    \label{eq:approx:hankel}
    H[p, s] = f_R(p s), \quad \text{for all } p \in P, s \in S.
    \end{equation}
    Constructing \(H\) requires evaluating the RNN on all concatenations \(ps\) where \(|p|, |s| \leq \theta\).
    \item \textbf{Compute Low-Rank Approximation:} Compute a rank-\(r\) truncated Singular Value Decomposition (SVD) of the Hankel matrix:
    \[
    H \approx U_r \Sigma_r V_r^\top,
    \]
    where \(r\) is the target dimension of the WA. The matrices \(U_r\) and \(V_r\) provide orthonormal bases for the column and row spaces of \(H\), respectively, and \(\Sigma_r\) contains the top \(r\) singular values.
    \item \textbf{Extract WA Parameters:} The WA parameters are derived by projecting the Hankel matrix and its shifts into the low-dimensional space defined by the SVD.
    \begin{itemize}
        \item \textbf{Initial and Final Vectors:} Let \(e_\epsilon\) be the basis vector corresponding to the empty string. The initial vector \(\nu\) and final vector \(\tau\) are defined as:

\begin{equation}
\label{eq:approx:initial-final}
        \nu = U_r \Sigma_r^{1/2} e_{\epsilon}, \quad \tau = (\Sigma_r^{1/2} V_r^\top)^\top e_{\epsilon}.
\end{equation}
        \item \textbf{Transition Matrices:} For each symbol \(\sigma \in \Sigma\), form the shifted Hankel matrix \(H_\sigma\), where:
        \begin{equation}
        \label{eq:approx:shifted-hankel}
        (H_\sigma)[p, s] = f_R(p \sigma s).
        \end{equation}
        The transition matrix \(M_\sigma\) for symbol \(\sigma\) is then obtained by projection:
        \begin{equation}
        \label{eq:approx:project}
        M_\sigma = \Sigma_r^{-1/2} U_r^\top H_\sigma V_r \Sigma_r^{-1/2}.
        \end{equation}
    \end{itemize}
\end{enumerate}

\paragraph{Outcome.}
The resulting automaton \(A = (\nu, \{M_\sigma\}_{\sigma \in \Sigma}, \tau)\) approximates the behavior of the RNN \(R\) on sequences up to length \(\theta\). If the Hankel matrix \(H\) has exact rank \(r\), the recovery by the WA is exact for all \(w \in \Sigma^{\leq \theta}\). Otherwise, the approximation error is bounded by the spectral error of the truncated SVD \cite{davis1970rotation}.

\begin{algorithm}[t]
\caption{\textsc{Convert-RNN-to-WA}($R$, $\theta$, $r$)}
\label{alg:rnn-to-wa}
\KwIn{RNN $R$, truncation length $\theta$, target rank $r$.}
\KwOut{WA $A=(\nu, \{M_\sigma\}, \tau)$ of dimension $r$.}
\BlankLine
Set $P \gets \Sigma^{\leq \theta}$, $S \gets \Sigma^{\leq \theta}$ and index them arbitrarily.\\
\For{$p \in P, s \in S$}{
    $H[p, s] \gets f_R(ps)$ \tcp*{Construct Hankel matrix (Eq. \ref{eq:approx:hankel})}
}
Compute rank-$r$ truncated SVD: $H \approx U_r \Sigma_r V_r^\top$\\
Define vectors $\nu \gets U_r \Sigma_r^{1/2} e_{\epsilon}$ and $\tau \gets (\Sigma_r^{1/2} V_r^\top)^\top e_{\epsilon}$ \tcp*{(Eq. \ref{eq:approx:initial-final})}
\For{$\sigma \in \Sigma$}{
    \For{$p \in P, s \in S$}{
        $H_\sigma[p, s] \gets f_R(p \sigma s)$ \tcp*{Construct shifted matrix (Eq. \ref{eq:approx:shifted-hankel})}
    }
    $M_\sigma \gets \Sigma_r^{-1/2} U_r^\top H_\sigma V_r \Sigma_r^{-1/2}$ \tcp*{Project to get transitions (Eq. \ref{eq:approx:project})}
}
\Return $A = (\nu, \{M_\sigma\}, \tau)$
\end{algorithm}

\subsection{Correctness and Analysis of the Reduction}

We now establish the correctness of the conversion procedure outlined in Algorithm~\ref{alg:rnn-to-wa}. The following theorem characterizes the approximation guarantees, separating errors arising from truncation and those from spectral approximation.

\begin{theorem}[Approximation Guarantees]
\label{thm:convert-correct}
Let \(R\) be an RNN with induced function \(f_R: \Sigma^* \to \mathbb{R}\). Fix a truncation length \(\theta > 0\) and let \(f_{R,\leq\theta}\) be its restriction to \(\Sigma^{\leq\theta}\). Let \(H\) be the Hankel matrix indexed by \(P = S = \Sigma^{\leq\theta}\) and let \(A\) be the WA returned by Algorithm~\ref{alg:rnn-to-wa} with target rank \(r\). Then:
\begin{enumerate}
    \item \textbf{Exact Recovery:} If \(\mathrm{rank}(H) = r\) and the SVD is computed exactly, then \(f_A(w) = f_R(w)\) for all \(w \in \Sigma^{\leq\theta}\).
    \item \textbf{Approximate Recovery:} If \(H\) is approximated by a rank-\(r\) matrix \(\widehat{H}\) such that \(\|H - \widehat{H}\|_2 \leq \eta\), then for every \(w \in \Sigma^{\leq\theta}\), the pointwise error is bounded by:
    \begin{equation}
    \label{eq:approx:approx-bound}
    |f_A(w) - f_R(w)| \leq C(\theta, |\Sigma|, r) \cdot \eta,
    \end{equation}
    where \(C(\theta, |\Sigma|, r)\) is a polynomial factor dependent on the truncation length, alphabet size, and target rank.
    \item \textbf{Full-Space Guarantee:} For a target distribution \(Q\) over \(\Sigma^*\), if the tail mass beyond length \(\theta\) is bounded (\(\sum_{|w|>\theta} Q(w) \leq \delta\)) and the approximation error on \(\Sigma^{\leq\theta}\) is bounded (\(\|f_A - f_R\|_1 \leq \gamma\)), then the WA \(A\) satisfies:
    \begin{equation}
    \label{eq:approx:full-bound}
    \|f_A - Q\|_1 \leq \|f_R - Q\|_1 + \gamma + 2\delta.
    \end{equation}
    If the RNN itself is an approximation to \(Q\) (i.e., \(\|f_R - Q\|_1 \leq \epsilon\)), then the total error is \(\epsilon + \gamma + 2\delta\).
\end{enumerate}
\end{theorem}

\begin{proof}[Proof Sketch]
(1) The exact recovery case is a direct consequence of the spectral learning theory for WFA \cite{balle2012spectral,balle2020learning}. When \(\mathrm{rank}(H) = r\), the truncated SVD yields a minimal linear representation for the function \(f_{R,\leq\theta}\), which is exactly captured by the constructed WA \(A\).

(2) The approximate case follows from perturbation analysis. The error \(\eta\) in the Hankel matrix approximation propagates to the parameters \((\nu, \{M_\sigma\}, \tau)\) of the WA. Standard bounds on the sensitivity of the SVD \cite{davis1970rotation} imply that this propagation is controlled, leading to a pointwise error that is polynomial in the problem parameters and linear in \(\eta\).

(3) The full-space guarantee is obtained by a decomposition of the \(\ell_1\) error:
\begin{align*}
\|f_A - Q\|_1 &= \sum_{w \in \Sigma^*} |f_A(w) - Q(w)|  \\
&\leq \underbrace{\sum_{|w| \leq \theta} |f_A(w) - f_R(w)|}_{\leq \gamma} + \underbrace{\sum_{|w| \leq \theta} |f_R(w) - Q(w)|}_{\leq \epsilon} + \underbrace{2 \sum_{|w| > \theta} Q(w)}_{\leq 2\delta}.
\end{align*}

The result follows from applying the bounds for each term.
\end{proof}

\paragraph{Complexity Analysis.}
The computational complexity of Algorithm~\ref{alg:rnn-to-wa} is dominated by the construction of the Hankel matrices and the SVD computation. Let \(k = |\Sigma^{\leq \theta}| = O(|\Sigma|^\theta)\). The algorithm requires:
\begin{itemize}
    \item \(O(k^2)\) evaluations of \(f_R\), each requiring \(O(\theta \cdot T_{\text{step}})\) time, where \(T_{\text{step}}\) is the cost of one RNN transition.
    \item \(O(k^3)\) time for the SVD of a \(k \times k\) matrix.
    \item \(O(|\Sigma| \cdot k^2)\) memory for storing the Hankel and shifted Hankel matrices.
\end{itemize}
This exponential dependence on \(\theta\) limits practical application to small truncation lengths, though sampling techniques can mitigate this cost. 
\cite{HsuKakadeZhang2009}
\paragraph{Remarks.}
This reduction demonstrates that the function computed by the RNN (and hence, by Theorem~\ref{thm:equivalence}, the self-attention mechanism) can be approximated by a WFA. The quality of approximation is governed by the truncation length $\theta$ and the spectral rank \(r\) of the finite Hankel matrix. This provides a foundational link between transformers and automata theory, enabling the application of formal methods to analyse their behaviour.
\section{Identity Testing for Rational Stochastic Languages}
\label{sec:testing}

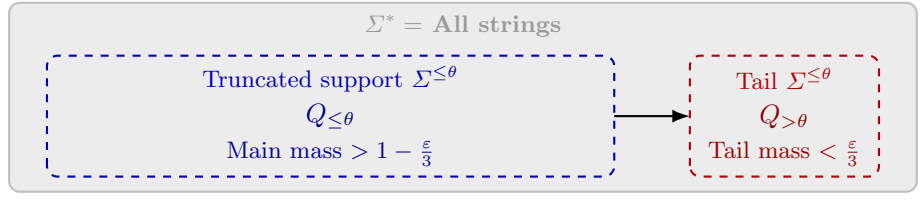
\begin{figure}[t]
    \centering
\begin{tikzpicture}[scale=1.0, >=Latex, thick]
\draw [fill=gray!15, draw=gray!50, rounded corners] (0.5,-0.2) rectangle (12.5,2.3);
\node at (6.5,2) [font=\small\bfseries, color=gray!80] {$\Sigma^* =$ All strings};

\draw [draw=blue!70!black, dashed, rounded corners, thick] (1,0) rectangle (8.5,1.6);
\node at (4.75,1.3) [font=\small, color=blue!70!black] {Truncated support $\Sigma^{\le \theta}$};
\node at (4.75,0.8) [font=\normalsize\bfseries, color=blue!55!black] {$Q_{\le \theta}$};
\node at (4.75,0.3) [font=\footnotesize, color=blue!60!black] {Main mass $> 1 - \frac{\varepsilon}{3}$};

\draw [draw=red!70!black, dashed, rounded corners, thick] (9.5,0) rectangle (12,1.6);
\node at (10.75,1.3) [font=\small, color=red!70!black] {Tail $\Sigma^{\le \theta}$};
\node at (10.75,0.8) [font=\normalsize\bfseries, color=red!55!black] {$Q_{>\theta}$};
\node at (10.75,0.3) [font=\footnotesize, color=red!60!black] {Tail mass $< \frac{\varepsilon}{3}$};

\draw [->, thick, black] (8.5,0.8) -- (9.5,0.8);

\end{tikzpicture}
    \caption{
Schematic of truncation and tail error decomposition. The gray rectangle represents the full domain $\Sigma^*$. 
The blue box is the truncated support $\Sigma^{\le \theta}$, containing most of the probability mass ($Q_{\le \theta} > 1 - \varepsilon/3$). 
Truncation ensures the tail mass $Q_{>\theta}$ is bounded by $\varepsilon/3$, and hence contributes minimally. 
The total $\ell_1$ error is decomposed as 
$\gamma$ (estimation) $+ C\eta$ (spectral) $+ \delta$ (tail mass), 
highlighting that the bulk of the distribution is concentrated in a manageable region while the tail is small and controlled.
}
    \label{fig:testing}
\end{figure}

A fundamental verification question is whether the distribution $P$ generated by a black-box model (e.g., an LLM) is close to a known, well-understood reference distribution $Q$. For our purposes, $Q$ is a rational stochastic language, specified by a weighted automaton or a stochastic regular expression (SRE). This motivates the study of \emph{identity testing} (also known as goodness-of-fit testing) for distributions over infinite domains \cite{Canonne2022Price}.

In this section, we present an identity testing algorithm for rational stochastic languages, which, to our knowledge, is the first to tackle this problem for infinite domains by leveraging the specific structure of rational distributions. The algorithm reduces the infinite-domain problem to a finite one via a tailored truncation strategy (illustrated in Figure~\ref{fig:testing}), enabling the use of sample-optimal finite-domain testers.

\begin{definition}[Identity Testing for Stochastic Languages]
\label{def:identity}
Let $Q$ be a known rational stochastic language, and let \( P \) be an unknown distribution over \( \Sigma^* \) accessible only via i.i.d. samples. Given parameters \( \varepsilon > 0 \) and \( \delta \in (0,1) \), an \emph{identity tester} is an algorithm that, with probability at least \( 1 - \delta \), outputs:
\begin{itemize}
    \item[$\bullet$] \textup{\textsc{Accept}} if \( \|P - Q\|_1 < \varepsilon \), and
    \item[$\bullet$] \textup{\textsc{Reject}} if \( \|P - Q\|_1 > \varepsilon \).
\end{itemize}
The sample complexity of the tester is the number of samples required from \( P \) to achieve this guarantee.
\end{definition}

We employ the \(\ell_1\) distance as it corresponds to the total variation distance, providing a robust and interpretable measure of discrepancy across the entire support. This is preferable to the \(\ell_\infty\) distance, which is very sensitive to pointwise errors, and the KL divergence, which requires absolute continuity.

\subsection{Reduction to Finite Support via Truncation}
\label{subsec:reduction}

The core of our approach is to truncate the distributions to strings of bounded length, ensuring the tail mass of the reference distribution \( Q \) is negligible. This reduction allows us to apply finite-domain identity testers.

\begin{definition}[Truncation Length $\theta$]
For a distribution $Q$, define $\theta(Q, \xi)$ as:

\[
\theta(Q, \xi) = \min_{\semantics{\mathcal{A}} = Q} \left( |\mathcal{A}_Q| \cdot \max\left\{1, \left\lceil \frac{\log(\xi)}{\log(1 - \rho(\mathcal{A}))} \right\rceil\right\} \right)
\]
where the minimum is taken over all equivalent WFAs $\mathcal{A}$ that represent $Q$, $|\mathcal{A_Q}|$ is the number of states, and $\rho(\mathcal{A})$ is the spectral radius (maximum eigenvalue) of the automaton's transition matrix.
\end{definition}

\begin{lemma}
    Let \( \varepsilon > 0 \) be the target accuracy. For the reference distribution \( Q \), the truncation length \( \theta = \theta(Q, \varepsilon/3) \) satisfies:
\[
Q_{> \theta} := \sum_{x \in \Sigma^{> \theta}} Q(x) < \frac{\varepsilon}{3}.
\]
\end{lemma}

\begin{proof}
Let $\mathcal{A}$ be a WFA realizing $Q$ that achieves the minimum in the definition of $\theta(Q, \varepsilon/3)$. Let $n$ be its number of states and $\rho$ its spectral radius.

By the spectral properties of WFAs, the probability that a string generated by $\mathcal{A}$ has length greater than $L$ decays roughly as $O(\rho^L)$ in the sense of total mass. More formally, there exists a constant $C$ such that for all $L \ge n$,
\[
Q_{> L} \le C \cdot \rho^{L - n + 1}.
\]

Choosing $L = \theta$ as defined, we have $\theta \ge n \cdot \max\left\{1, \left\lceil \frac{\log(\varepsilon/3)}{\log(1 - \rho)} \right\rceil\right\}$. In particular, $\theta \ge n \cdot \left\lceil \frac{\log(\varepsilon/3)}{\log(1 - \rho)} \right\rceil$ when $\rho < 1$.

Using the inequality $\log(1 - \rho) \le -\rho$, we get
\[
\left\lceil \frac{\log(\varepsilon/3)}{\log(1 - \rho)} \right\rceil \ge \frac{\log(1/(\varepsilon/3))}{\rho} = \frac{\log(3/\varepsilon)}{\rho}.
\]

Thus $\theta \ge \frac{n \log(3/\varepsilon)}{\rho}$, so $\rho^\theta \le e^{-n \log(3/\varepsilon)} = (\varepsilon/3)^n$.

Hence $Q_{> \theta} \le C \cdot (\varepsilon/3)^n \cdot \rho^{-n+1}$. For sufficiently small $\varepsilon$, this is less than $\varepsilon/3$ (adjusting constants if needed by the minimization over WFAs).
\end{proof}

\begin{algorithm}[t]
\caption{$\ell_1$-\textsc{IdentityTester} for Rational Stochastic Languages}
\label{alg:l1-identity-tester}
\KwIn{Known distribution $Q$, sample access to $P$, parameters $\varepsilon > 0$, $\delta \in (0,1)$}
\KwOut{\textup{\textsc{Accept}} if $\|P - Q\|_1 < \varepsilon$; \textup{\textsc{Reject}} if $\|P - Q\|_1 > \varepsilon$ (w.p. $\geq 1-\delta$)}
\BlankLine
Compute truncation length $\theta(Q, \varepsilon/3)$.\\
Let $k \gets |\Sigma^{\leq \theta}| = \Theta(|\Sigma|^{\theta+1})$.\\
Draw $N = \Theta\left( \frac{\sqrt{k}}{\varepsilon^2} + \frac{k}{\log k} \cdot \log(\frac{1}{\delta}) \right)$ i.i.d. samples from $P$.\\
Let $\eta$ be the fraction of samples with $|x| > \theta$.\\
\If{$\eta > \varepsilon/3$}{
    \Return \textup{\textsc{Reject}} \tcp*{Empirical tail mass is too large}
}
Let $\widehat{P}_{\le \theta}$ be the empirical distribution of the remaining samples over $\Sigma^{\le \theta}$.\\
Let $Q_{\le \theta}$ be the restriction of $Q$ to $\Sigma^{\le \theta}$.\\
Run a finite-domain \emph{tolerant} identity tester on $\widehat{P}_{\le \theta}$ and $Q_{\le \theta}$ with distance thresholds $(\varepsilon/3, \varepsilon)$ and confidence $1-\delta$.\\
\Return the output of the finite-domain tester.
\end{algorithm}

\subsection{Correctness and Sample Complexity Analysis}
\label{subsec:correctness}

We now establish the formal guarantees for Algorithm~\ref{alg:l1-identity-tester}.

\begin{theorem}[Correctness of Identity Tester]
\label{thm:correctness}
Let \( Q \in \Stoch^{\mathrm{rat}}_{\mathbb{R}^+}(\Sigma^*) \) be given as an SRE, and let \( P \) be an unknown distribution over \( \Sigma^* \). For any \( \varepsilon > 0 \) and \( \delta \in (0,1) \), Algorithm~\ref{alg:l1-identity-tester} satisfies the following with probability at least \( 1 - \delta \):
\begin{itemize}
    \item[$\bullet$] \textbf{(Completeness)} If \( \|P - Q\|_1 < \varepsilon \), then output is \textsc{Accept}.
    \item[$\bullet$] \textbf{(Soundness)} If \( \|P - Q\|_1 > \varepsilon \), then output is \textsc{Reject}.
\end{itemize}
The sample complexity is:
\[
N = \Theta\left( \frac{\sqrt{k}}{\varepsilon^2} + \frac{k}{\log k} \cdot \log\left(\frac{1}{\delta}\right) \right), \quad \text{where } k = |\Sigma^{\leq \theta}|.
\]
\end{theorem}

\begin{proof}[Sketch]
The proof hinges on the decomposition of the \(\ell_1\) error and the properties of the finite-domain tolerant tester.
\begin{enumerate}
    \item \textbf{Truncation Guarantee:} By construction, $Q_{> \theta} < \varepsilon/3$.
    \item \textbf{Tail Mass Check:} If $P_{> \theta} > \varepsilon/3$, then the empirical estimate $\eta$ will exceed $\varepsilon/3$ with high probability given the sample size $N$, causing rejection. This handles soundness cases where $P$ diverges from $Q$ in the tail.
    \item \textbf{Head Analysis:} If the algorithm proceeds to the finite-domain tester, we know $P_{> \theta} \leq \varepsilon/3 + o(1)$ with high probability. The total \(\ell_1\) distance can be bounded as:
    \[
    \|P - Q\|_1 \leq \underbrace{\|P_{\leq \theta} - Q_{\leq \theta}\|_1}_{\text{Head}} + \underbrace{P_{> \theta} + Q_{> \theta}}_{\text{Tail}}.
    \]
    \item \textbf{Completeness ($\|P - Q\|_1 < \varepsilon$):} The tail term is less than $\varepsilon/3 + \varepsilon/3 = 2\varepsilon/3$. Thus, $\|P_{\leq \theta} - Q_{\leq \theta}\|_1 < \varepsilon/3$. The tolerant tester, configured with a lower threshold of $\varepsilon/3$, will correctly accept.
    \item \textbf{Soundness ($\|P - Q\|_1 > \varepsilon$):} Since the tail term is at most $2\varepsilon/3$, the head distance must satisfy $\|P_{\leq \theta} - Q_{\leq \theta}\|_1 > \varepsilon/3$. The tolerant tester, with an upper threshold of $\varepsilon$, will correctly reject.
\end{enumerate}
The sample complexity is dominated by the finite-domain tolerant identity tester \cite{Canonne2022Price}, which requires $O(\sqrt{k}/\varepsilon^2 + k/\log k)$ samples to distinguish between distances of $\varepsilon/3$ and $\varepsilon$ on a domain of size $k$. The multiplicative $\log(1/\delta)$ factor comes from standard confidence amplification. The number of samples is sufficient to ensure the empirical tail mass check is also correct with probability $1-\delta$.
\end{proof}

\paragraph{Remarks.}
The algorithm provides a computationally efficient reduction from identity testing on infinite sequences to the finite domain. The sample complexity, while exponential in the truncation length \(\theta\), is polynomial in the effective support size \(k = |\Sigma^{\leq \theta}|\). For rational stochastic languages with rapidly decaying tails (a common property), \(\theta\) is manageable, leading to a practical verification procedure for models that are hypothesized to be close to such languages.
\section{Conclusion and Future Work}
\label{sec:conclusion}

We have established the first formal reduction from decoder-only transformers to weighted automata, creating a bridge between modern neural sequence models and classical automata theory. Our main technical contributions are: (1) a constructive proof that transformer self-attention can be computed by structured RNNs, (2) a spectral learning algorithm that converts such RNNs to weighted automata with provable approximation guarantees, and (3) an identity testing procedure for rational stochastic languages with sample complexity $O(\sqrt{k}/\varepsilon^2 + k \log k)$ where $k = |\Sigma^{\leq \theta}|$.

\paragraph{The Computational Barrier.} Our analysis reveals a fundamental bottleneck: spectral learning requires Hankel matrices of size $ O(|\Sigma|^{2\theta})$, leading to exponential complexity in the truncation parameter $\theta$. This barrier underscores why automata-based verification of practical language models has remained elusive.

Importantly, this limitation is not merely an artifact of our specific construction. The exponential dependence on $\theta$ reflects a deeper tension between the expressivity required to approximate transformer behaviour and the computational tractability of automata-theoretic methods. Any approach that seeks to capture the distributional properties of transformers over realistic vocabularies must confront this fundamental trade-off.

\paragraph{Theoretical Implications.} Our results locate transformers within the established hierarchy of computational models. The reduction provides the first formal characterization of transformer expressivity in terms of weighted automata. This connection enables transfer of results: closure properties, decidability results, and algorithmic techniques from automata theory now apply to transformer analysis. Furthermore, our identity testing algorithm provides a principled statistical framework for model comparison and verification.

\paragraph{Algorithmic Challenges and Directions.} The exponential bottleneck suggests several concrete research directions:

\begin{enumerate}
    \item Develop low-rank approximations or sketching techniques for learning that preserve essential spectral properties while reducing computational cost.
    \item Design randomized algorithms that estimate automaton parameters without constructing full Hankel matrices, potentially achieving polynomial dependence on $\theta$.
    \item  Exploit specific properties of transformer architectures (attention patterns, positional encodings, etc.) to reduce the effective dimension of the learning problem, and make it tractable. 
\end{enumerate}

We acknowledge that these challenges have not allowed us to implement and evaluate our approach on real-world transformers, and it remains a direction of future work. Beyond this, Several fundamental questions remain: Can the exponential dependence on $\theta$ be avoided while preserving approximation quality? What is the minimal class of weighted automata needed to approximate practical transformers? How do architectural choices (depth, attention patterns, normalization) affect the complexity of automata-based approximation?

In conclusion, our framework has provided new evaluation methodologies for language models. Rather than relying solely on empirical benchmarks, our identity testing approach provides statistical certificates of distributional similarity to reference models—a step toward more rigorous model validation. Through this work, we provide both a foundational connection between transformers and automata theory and a precise characterization of the computational challenges that must be overcome for practical verification. The path forward requires algorithmic innovation at the intersection of spectral methods, randomized algorithms, and formal language theory.

\bibliographystyle{splncs04}
\bibliography{references}

\end{document}